\documentclass[11pt]{article}

\usepackage{amsfonts}
\usepackage{amsmath}
\usepackage{amssymb}
\usepackage{amsthm}
\usepackage{hyperref}

\usepackage{graphicx}

\theoremstyle{plain}
\newtheorem{definition}{Definition}
\newtheorem{theorem}{Theorem}

\theoremstyle{definition}
\newtheorem{example}{Example}
\newtheorem{remark}{Remark}

\DeclareMathOperator{\Tr}{tr}

\DeclareMathOperator{\Int}{int}

\newcommand{\Vect}[1]{{\mathbf{#1}}}

\renewcommand{\Re}{\mathop{\mathrm{Re}}\nolimits}

\begin{document}

\title{A note on the replicator equation with explicit\\ space and global regulation}

\author{Alexander S. Bratus'$^1$, Vladimir P. Posvyanskii$^1$, Artem S. Novozhilov$^{1,}$\footnote{Corresponding author: anovozhilov@gmail.com}\\[3mm]
\textit{\normalsize $^{1}$Applied Mathematics--1, Moscow State University of Railway Engineering,}\\[-1mm]\textit{\normalsize Obraztsova 9, Moscow 127994, Russia}}

\date{}

\maketitle

\begin{abstract}
A replicator equation with explicit space and global regulation is considered. This model provides a natural
framework to follow frequencies of species that are distributed in the space. For this model, analogues to classical notions of the Nash equilibrium and evolutionary stable strategy are provided. A sufficient condition for a uniform stationary state to be a spatially distributed evolutionary stable state is presented and illustrated with examples.

\end{abstract}

\section{Preliminaries and notation}\label{sec1}
The general replicator equation comprises well-established biomathematical models that arise in quite distinct evolutionary contexts (see, e.g., \cite{hofbauer1998ega,hofbauer2003egd,schuster1983rd}). In particular, this equation appears in the areas of theoretical population genetics (e.g., \cite{hofbauer1998ega,svirezhev:fme}), prebiotic molecular evolution \cite{bratus2006ssc,bratus2009existence,Eigen1977}, and evolutionary game theory \cite{hofbauer2003egd,smith1982evolution,taylor1978ess}.

Arguably, one of the simplest \textit{replicator equations} takes the form
\begin{equation}\label{eq1:1}
    \dot{v}_i=v_i\left[(\Vect{Av})_i-{f}^{loc}(t)\right],\qquad i=1,\ldots,n.
\end{equation}
Here $\Vect{v}=\Vect{v}(t)=(v_1(t),\ldots,v_n(t))\in \mathbb R^n$ is a vector-function of $n$ variables, $\Vect{A}$ is a constant $n\times n$ matrix with elements $a_{ij}\in\mathbb R$, $(\Vect{Av})_i$ is the $i$-th element of the vector $\Vect{Av}$, $(\Vect{Av})_i=\sum_{i=1}^na_{ij}v_j(t)$, ${f}^{loc}(t)$ is a function, which is determined later, and dot, as usual, is used to denote differentiation with respect to the time variable $t$. It is customarily  supposed that $v_i(t)$ describes a relative abundance of the $i$-th species such that the total concentration $\sum_{i=1}^nv_i(t)$ is kept constant and often equal, without loss of generality, to 1. Therefore, the state space of \eqref{eq1:1} is the simplex $S_n=\{\Vect{v}\colon \sum_{i=1}^nv_i(t)=1,\,v_i(t)\geq 0,\,i=1,\ldots,n\}$, which is invariant under \eqref{eq1:1} if we set ${f}^{loc}(t)=\langle \Vect{Av},\,\Vect{v}\rangle=\sum_{i=1}^n(\Vect{Av})_iv_i$, so that $\langle\cdot\,,\cdot\rangle$ denotes the usual scalar product in $\mathbb R^n$. Since $(\Vect{Av})_i$ gives the net rate of growth of the $i$-th species in our system, the quantity $(\Vect{Av})_i$ is termed as the Malthusian fitness, and hence ${f}^{loc}(t)$ represents the average fitness of the population at time $t$.

System \eqref{eq1:1} is a very well studied object, see, e.g., \cite{cressman2003evolutionary,hofbauer2003egd,hofbauer1998ega}. Very briefly, the rest points of \eqref{eq1:1} are given by the solutions to the following system
\begin{equation}\label{eq1:2}
      (\Vect{Av})_1 =(\Vect{Av})_2=\ldots=(\Vect{Av})_n=\langle\Vect{Av},\,\Vect{v}\rangle=\beta, \quad
        \Vect{v} \in  S_n.
\end{equation}
In general, \eqref{eq1:2} can have no solutions in $\Int S_n=\{\Vect{v}\colon \sum_{i=1}^nv_i=1,\,v_i>0,\,i=1,\ldots,n\}$, a unique solution, or infinitely many solutions. The necessary condition to have a solution to \eqref{eq1:2} is the linear dependence of the rows of $\Vect{A}$, which we denote $\textbf{A}_i$, and the vector $\Vect{1}_n=(1,1,\ldots,1)$:
$$
\mu_0\Vect{1}_n+\sum_{i=1}^n\mu_i\Vect{A}_i=0,\quad \mu_0\neq 0,\quad \sum_{i=1}^n\mu_i\neq 0.
$$

A natural way to derive the replicator equation \eqref{eq1:1} from the first principles is to start with a \textit{selection system} of the form
\begin{equation}\label{eq1:1a}
\dot{y}_i=F_i(\Vect{v})y_i,\qquad i=1,\ldots,n,
\end{equation}
where $\Vect{y}=(y_1(t),\ldots,y_n(t))\in\mathbb R^n_+$ is a vector of absolute sizes, and $F_i(\Vect{v})$ denotes the Malthusian fitness (the per capita birth rate) of species $y_i$, which can depend on the structure of the total population at the time $t$. Assuming that $\Vect{y}$ does not tend to zero, the change of the variables $v_i=y_i/(\sum_{i=1}^ny_i)$ leads to the replicator equation \eqref{eq1:1}, if $F_i(\Vect{v})=(\Vect{Av})_i$. Therefore, it is equivalent to study the selection system \eqref{eq1:1a} or the replicator equation \eqref{eq1:1} (see \cite{karev2009,karev2010} for more details).

As it was mentioned, the replicator equation \eqref{eq1:1} naturally arises in the evolutionary game theory (for the origin, see \cite{smith1982evolution,smith1973logic}). There exists a parallel between the concepts of the game theory with a payoff matrix $\Vect{A}$ and the behavior of the solutions to the replicator equation \eqref{eq1:1} \cite{hofbauer1998ega}. In particular, one of the central notions of the game theory, the \textit{Nash equilibrium} $\Vect{\hat{v}}$, is defined as such $\Vect{\hat{v}}\in S_n$ for which
\begin{equation}\label{eq1:3}
    \langle\Vect{v},\,\Vect{A\hat{v}}\rangle\leq \langle\Vect{\hat{v},\,\Vect{A\hat{v}}}\rangle
\end{equation}
for any $\Vect{v}\in S_n$; and an \textit{evolutionary stable state} (ESS) $\Vect{\hat{v}}\in S_n$ is defined as
\begin{equation}\label{eq1:4}
    \langle\Vect{\hat{v}},\,\Vect{A{v}}\rangle> \langle\Vect{{v},\,\Vect{A{v}}}\rangle
\end{equation}
for any $\Vect{v}\neq\Vect{\hat{v}}$ in a neighborhood of $\Vect{\hat{v}}\in S_n$.

It can be shown (see \cite{hofbauer1998ega} for the details and proofs) that if $\Vect{\hat{v}}$ is a Nash equilibrium of the game with payoff matrix $\Vect{A}$, then $\Vect{\hat{v}}$ is a rest point of \eqref{eq1:1}. Moreover, if $\Vect{\hat{v}}$ is a rest point of \eqref{eq1:1} and Lyapunov stable, then it is a Nash equilibrium; and if $\Vect{\hat{v}}$ is ESS, then it is an asymptotically stable rest point of \eqref{eq1:1}.

We remark that model \eqref{eq1:1} is a system of ordinary differential equations (ODE), i.e., it is a mean-field model. A significant attention was drawn to the replicator equation in the case when heterogeneous spatial structure can be included into the model formulation \cite{dieckmann2000}. One of the suggested solution was spatially explicit models (see \cite{boerlijst1990ssa,boerlijst1991sws,cronhjort1994hvp} for the models of molecular evolution). In general, there are several different approaches to include spatial structure into the replicator equation. The solution to the problem when all the diffusion rates are equal is straightforward: in this case, following the ecological approach, we can just add the Laplace operator to the right hand sides of \eqref{eq1:1} (this was used, e.g., in \cite{fisher1937waa,hadeler1981dfs}). However, the assumption of the equal diffusion rates would be too stringent in the general situation. To overcome this problem, Vickers et al. \cite{cressman2003evolutionary,hutson1995sst,vickers1989spa} introduced a special form of the population regulation to allow for different diffusion rates. In these works a nonlinear term is used that provides a \textit{local} regulation of the populations under question to keep the total population size constant, although no particular biological mechanism is known that lets individuals adapt their per capita birth and death rates to local circumstances \cite{ferriere2000ads}. A more straightforward approach, in our view, would be to start with a spatially explicit selection system of the form
$$
\frac{\partial y_i}{\partial t}=F_i(\Vect{v})y_i+d_i\Delta y_i,\qquad i=1,\ldots,n
$$
and apply transformation $v_i=y_i/(\sum_{i=1}^n\int_{\Omega} y_i\,d\Vect{x})$. Here $y_i=y_i(\Vect{x},t),\,\Vect{x}\in\Omega\subset \mathbb R^m$, and $d_i>0,\,i=1,\ldots,n,$ are diffusion coefficients. This approach is similar to the lines how the replicator equation \eqref{eq1:1} can be obtained from \eqref{eq1:1a}. In this way, assuming impenetrable boundary of the area $\Omega$,  we automatically obtain the condition of the \textit{global} population regulation
$$
\sum_{i=1}^n\int_{\Omega}v_i(\Vect{x},t)\,d\Vect{x}=1,
$$
which was considered in some earlier works on the mathematical models of the prebiotic molecular evolution \cite{bratus2009stability,bnp2010,bratus2006ssc,bratus2009existence,weinberger1991ssa}. Here we extend this approach to the general replicator equation.

The rest of the paper is organized as follows. In Section \ref{sec2} the model formulation is presented, together with some additional definitions and notation. Section \ref{sec3} is devoted to stability analysis of spatially homogeneous equilibria of the distributed replicator equation. In Section \ref{sec4} we formulate possible extensions of the notions of the Nash equilibrium and ESS for our spatially explicit model and present some consequences of the new definitions. In Section \ref{sec5} we derive sufficient conditions for a distributed ESS along with some illustrative examples.

\section{Replicator equation with explicit space}\label{sec2} Let $\Omega$ be a bounded domain, $\Omega\in \mathbb R^m,\,m=1,\,2,$ or $3$, with a piecewise-smooth boundary $\Gamma$. In the following we assume, without loss of generality, that the volume of $\Omega$ is equal to 1, i.e., $\int_{\Omega}\,d\textbf{x}$=1. A spatially explicit analogue to \eqref{eq1:1} is given by the following reaction-diffusion system
\begin{equation}\label{eq2:1}
    \partial_tu_i=u_i \left[(\Vect{Au})_i-{f}^{sp}(t)\right]+d_i\Delta u_i,\quad i=1,\ldots,n,\quad t>0.
\end{equation}
Here $u_i=u_i(\textbf{x},t),\,\textbf{x}\in\Omega,\,\partial_t=\frac{\partial}{\partial t}\,,\,\Delta$ is the Laplace operator, in Cartesian coordinates $\Delta=\sum_{k=1}^m\frac{\partial^2}{\partial x_k^2}\,,\,\Vect{A}$ is a given constant matrix, and $d_i>0,\,i=1,\ldots,n$ are the diffusion coefficients. The initial conditions are $u_i(x,0)=\varphi_i(x),\,i=1,\ldots,n,\,$ and the form of ${f}^{sp}(t)$ will be determined later.

It is natural to assume that we consider closed systems (see also
\cite{weinberger1991ssa}), i.e., we have the boundary conditions
\begin{equation}\label{eq2:2}
    \partial_{\textbf{n}}u_i|_{\textbf{x}\in\Gamma}=\left.\frac{\partial u_i(\textbf{x},t)}{\partial
    \textbf{n}}\right|_{\textbf{x}\in\Gamma}=0,\quad t>0,\quad i=1,\ldots,n,
\end{equation}
where $\textbf{n}$ is the normal vector to the boundary $\Gamma$.

As it was discussed in Section \ref{sec1}, the global regulation of the total
species concentrations occurs in the system, such that
\begin{equation}\label{eq2:3}
\sum_{i=1}^n \int_{\Omega}u_i(\textbf{x},t)\,d\textbf{x}=1
\end{equation}
for any time moment $t$. This condition is analogous to the condition
for the constant concentration in the finite-dimensional
case \cite{hofbauer1998ega}. From the boundary condition \eqref{eq2:2}
and the integral invariant \eqref{eq2:3} the expression for the
function ${f}^{sp}(t)$ follows:
\begin{equation}\label{eq2:4}
    {f}^{sp}(t)=\sum_{i=1}^n\int_{\Omega}u_i(\Vect{Au})_i\,d\textbf{x}=\int_{\Omega}\langle \Vect{Au},\,\Vect{u}\rangle\,d\textbf{x}.
\end{equation}

Suppose that for any fixed $\textbf{x}\in\Omega$ each function $u_i(\textbf{x},t)$ is
differentiable with respect to variable $t$, and belongs to the
Sobolev space $W^1_2(\Omega)$ if $m=1$ or $W_2^2(\Omega)$ if $m=2,3$ as the function of $\textbf{x}$ for any fixed
$t>0$. Here $W_{2}^k,\,k=1,2$  is the space of functions, which have square integrable derivatives with respect to $x\in \Omega$ up to the order $k$. Note, that from the embedding theorem (e.g., \cite{cantrell2003spatial,mikhlin1964vmm}) we have that any function from the space $W_2^k(\Omega)$ is continuous, except possibly on the set of measure zero, and this continuity is used in some proofs below.

Denote $\Omega_t=\Omega\times[0,\infty)$ and consider the space of functions $B(\Omega_t)$ with the norm
$$
\|{y}\|_B= \max_{t\geq 0}\left\{\|{y}(\textbf{x},t)\|_{W_2^k}+\|\frac{\partial {y}}{\partial t}(\textbf{x},t)\|_{W_2^k}\right\}.
$$

Hereinafter we denote $S_n(\Omega_t)$ the set of non-negative vector-functions $\Vect{u}(\textbf{x},t)$, $u_i(\textbf{x},t)\in B(\Omega_t),\,i=1,\ldots,n,$ which satisfy \eqref{eq2:3}, and we use the notation $\Int S_n(\Omega_t)$ for the set of functions $\Vect{u}(x,t)\in S_n(\Omega_t)$ for which $u_i(\textbf{x},t)>0,\,i=1,\ldots,n$.

We consider weak solutions to \eqref{eq2:1}, i.e., the solutions should satisfy the integral identity
$$
\int\limits_0^\infty\int\limits_\Omega\frac{\partial u_i}{\partial t}\eta\,d\textbf{x}dt=\int\limits_0^\infty\int\limits_\Omega u_i\left[(\Vect{Au})_i-{f}^{sp}(t)\right]\eta\,d\textbf{x}dt-d_i\int\limits_0^\infty\int\limits_\Omega\langle\nabla{u_i},\nabla\eta\rangle\,d\textbf{x}dt
$$
for any function $\eta=\eta(\textbf{x},t)$ on compact support, which is differentiable on $[0,\infty)$ with respect to $t$ and belongs to the Sobolev space $W_2^1(\Omega)$ for any fixed $t>0$.

Remark that generally system \eqref{eq2:1} is not a ``system of differential equations'' because its right-hand side contains functional \eqref{eq2:4}.

The steady state solutions to \eqref{eq2:1} can be found as the solutions to the following elliptic problem
\begin{equation}\label{eq2:5}
    d_i\Delta w_i+w_i\left[(\Vect{Aw})_i-{f^{sp}}\right]=0,\quad i=1,\ldots,n
\end{equation}
with the boundary conditions $\partial_\Vect{n}w_i(\textbf{x})=0$ on $\Gamma$. Here $w_i(\textbf{x})\in W_2^k(\Omega)$.

The integral invariant \eqref{eq2:3} now reads
\begin{equation}\label{eq2:6}
    \sum_{i=1}^n\int_{\Omega}w_i(\textbf{x})d\textbf{x}=1.
\end{equation}
The set of all non-negative vector-functions $\Vect{w}(\textbf{x})=(w_i(\textbf{x}),\ldots,w_n(\textbf{x})),\,w_i(\textbf{x})\in W_2^k(\Omega),\,i=1,\ldots,n,$ that satisfy \eqref{eq2:6} is denoted $S_n(\Omega)$. Using \eqref{eq2:6} and \eqref{eq2:5} we obtain that
\begin{equation}\label{eq2:7}
    f^{sp}=\int_\Omega\langle\Vect{Aw},\,\Vect{w}\rangle\,d\textbf{x},
\end{equation}
i.e., $f^{sp}$ is a constant. The rest points of \eqref{eq1:1} are spatially homogeneous solutions to \eqref{eq2:5}; the converse also holds: any spatially homogeneous solution to \eqref{eq2:5} is a rest point of \eqref{eq1:1}. In \cite{bratus2006ssc,bratus2009existence} we proved that for sufficiently small values of the diffusion coefficients $d_i$ there exist non-homogeneous solutions to \eqref{eq2:5} in the case $(\Vect{Au})_i=k_iu_i$ and $(\Vect{Au})_i=k_iu_{i-1}$ for arbitrary positive constants $k_i$ (these are autocatalytic and hypercyclic systems, respectively, for more details see \cite{bratus2006ssc,bratus2009existence}).

Here we undertake a task to investigate the stability of spatially homogeneous solutions to \eqref{eq2:1}, which can differ from the stability of the rest points of the local system \eqref{eq1:1} due to the explicit space in the model. We also consider the conditions necessary for spatially non-homogeneous solutions to appear. And yet another purpose is to transfer the definitions of the Nash equilibria and ESS for the distributed system \eqref{eq2:1}-\eqref{eq2:3} and apply these concepts to study the asymptotic behavior of \eqref{eq2:1}.

\section{Stability of spatially homogeneous equilibria}\label{sec3}

We start with the standard definition that is  used extensively throughout the text.

\begin{definition}\label{def3:1}A stationary solution $\Vect{\hat{w}}(\Vect{x})\in S_n(\Omega)$ to \eqref{eq2:5} is Lyapunov stable if for any $\varepsilon>0$ there exists a neighborhood
$$
U^\delta=\bigl\{\Vect{w}(\Vect{x})\in S_n(\Omega) \colon \sum_{i=1}^n\|{\hat{w}_i}(\Vect{x})-{w}_i(\Vect{x})\|^2_{W_2^k(
\Omega)}<\delta^2\bigr\}
$$
of $\Vect{\hat{w}}(\Vect{x})$ that for any initial data of \eqref{eq2:1} from $U^\delta$ it follows that
\begin{equation}\label{eq3:1}
    \sum_{i=1}^n\|u_i(\Vect{x},t)-\hat{w}_i(\Vect{x})\|^2_{B(\Omega_t)}\leq \varepsilon ^2
\end{equation}
for any $t\geq 0$.

If in \eqref{eq3:1} the left hand side tends to zero, then $\Vect{\hat{w}}(\Vect{x})$ is asymptotically stable.
\end{definition}

In Definition \ref{def3:1} $u_i(\textbf{x},t),\,i=1,\ldots,n$, are the corresponding solutions to \eqref{eq2:1} with the initial data $w_i(\textbf{x})\in U^\delta,\,i=1,\ldots,n$.

Consider the following eigenvalue problem
\begin{equation}\label{eq3:2}
    \Delta \psi(\Vect{x})+\lambda \psi(\Vect{x})=0,\quad \Vect{x}\in\Omega,\quad \partial_\Vect{n}\psi|_{\textbf{x}\in\Gamma}=0.
\end{equation}
The eigenfunction system of \eqref{eq3:2} is given by $\psi_0(\textbf{x})=1,\,\{\psi_i(\textbf{x})\}_{i=1}^\infty$ and forms a complete system in the Sobolev space $W_2^2(\Omega)$ (e.g., \cite{mikhlin1964vmm}), additionally
\begin{equation}\label{eq3:3}
    \langle \psi_i(\textbf{x}),\,\psi_j(\textbf{x})\rangle=\int_\Omega \psi_i(\textbf{x})\psi_j(\textbf{x})\,d\textbf{x}=\delta_{ij},
\end{equation}
where $\delta_{ij}$ is the Kronecker symbol. The corresponding eigenvalues satisfy the condition
$$
0=\lambda_0<\lambda_1\leq \lambda_2\leq\ldots\leq\lambda_i\leq\ldots,\qquad \lim_{i\to\infty}\lambda_i=+\infty.
$$

The following theorem gives a necessary condition for a spatially homogeneous solution to \eqref{eq2:1} be asymptotically stable, if the corresponding rest point of \eqref{eq1:1} is asymptotically stable.

\begin{theorem}\label{th3:1} Let $\Vect{\hat{v}}\in\Int S_n$ be an asymptotically stable rest point of \eqref{eq1:1}. Then for this point to be an asymptotically stable homogeneous stationary solution to \eqref{eq2:1} it is necessary that
\begin{equation}\label{eq3:4}
    \sum_{i=1}^n d_i> \frac{\beta}{\lambda_1}\,,\quad \beta=\langle\Vect{A\hat{v}},\,\Vect{\hat{v}}\rangle,
\end{equation}
where $\lambda_1$ is the first non-zero eigenvalue of \eqref{eq3:2}.
\end{theorem}
\begin{proof}
Remark that if $\Vect{\hat{v}}$ is Lyapunov stable then $\Vect{\hat{v}}$ is a Nash equilibrium of the game with the payoff matrix $\Vect{A}$. Consider the solutions to \eqref{eq2:1}--\eqref{eq2:3} assuming that the Cauchy data are perturbed:
$$
u_i(\textbf{x},0)=\varphi_i(\textbf{x})=\hat{v}_i+w_i^0(\textbf{x}),\quad i=1,\ldots,n,
$$
where $w_i^0(\textbf{x})\in W_2^k(\Omega)$ such that $\sum_{i=1}^n\|w_i^0(\textbf{x})\|^2_{W_2^k(\Omega)}\leq \delta^2,\,\delta >0$.

Let us look for a solution to \eqref{eq2:1}--\eqref{eq2:3} in the form
\begin{equation}\label{eq3:5}
    u_i(\textbf{x},t)=\hat{v}_i+w_i(\textbf{x},t),\quad w_i(\textbf{x},t)=c_0^i(t)+\sum_{k=1}^\infty c_k^i(t)\psi_k(\textbf{x}),\quad i=1,\ldots,n,
\end{equation}
which is always possible since the eigenfunctions $\psi_i(\Vect{x})$ of \eqref{eq3:2} form a complete system in $W_2^k(\Omega)$. $c_k^i(t),\,k=1,2,3,\ldots,i=1,\ldots,n$ are smooth functions of $t$. Note that $w_i^0(\textbf{x})=w_i(\textbf{x},0),\,i=1,\ldots,n$. The spatially homogeneous equilibrium $\Vect{\hat{v}}$ is stable if $c_k^i(t)\to0,\,t\to\infty$ for all $i$ and~$k$.

Due to the fact that
$$
\sum_{i=1}^n \int_\Omega u_i(\textbf{x},t)\,dx=\sum_{i=1}^n \hat{v}_i=1,
$$
we have from \eqref{eq3:2} that
\begin{equation}\label{eq3:6}
    \sum_{i=1}^nc_0^i(t)=0.
\end{equation}
Substituting \eqref{eq3:5} into \eqref{eq2:1} and retaining in the usual way only linear terms with respect to $w_i$ we obtain the following equations:
\begin{equation}\label{eq3:7}
    \frac{dc_0^i(t)}{dt}=\hat{v}_i\left[(\Vect{A{c}_0})_i-\langle \Vect{A^\tau\hat{v},\,\Vect{{c}_0}}\rangle\right]-\hat{v}_i\langle \Vect{A\hat{v}},\,\Vect{{c}_0}\rangle+c_0^i(t)\left[(\Vect{A\hat{v}})_i-\langle \Vect{A\hat{v},\,\Vect{\hat{v}}}\rangle\right],
\end{equation}
for $i=1,\ldots,n$. Here $\Vect{{c}_0}=(c_0^1(t),\ldots,c_0^n(t))$, and $\tau$ denotes transposition. From \eqref{eq1:2} it follows that the last term in \eqref{eq3:7} is zero. From \eqref{eq1:3} and \eqref{eq3:6} we also have that
$$
\langle\Vect{A\hat{v}},\,\Vect{c_0}\rangle=\sum_{i=1}^n c_0^i(t)(\Vect{A\hat{v}})_i=\beta\sum_{i=1}^nc_0^i(t)=0.
$$
Therefore, we obtain that
\begin{equation}\label{eq3:8}
    \frac{dc_0^i(t)}{dt}=\hat{v}_i\left[(\Vect{A{c}_0})_i-\langle \Vect{A^\tau\hat{v},\,\Vect{{c}_0}}\rangle\right],\quad i=1,\ldots,n.
\end{equation}

System \eqref{eq3:8} is linear, with the matrix $\textbf{Q}=\|q_{ij}\|_{i,j=1,\ldots,n}$ with the elements
\begin{equation}\label{eq3:8a}
q_{ij}=a_{ij}\hat{v}_i-(\Vect{A^\tau\hat{v}})_i\hat{v}_i,\quad i,j=1,\ldots,n,
\end{equation}
which coincides with the Jacobi matrix of \eqref{eq1:1} at the rest point $\Vect{\hat{v}}$. Due to the assumptions the rest point of \eqref{eq1:1} is asymptotically stable, therefore the trivial stationary point of \eqref{eq3:8} is also asymptotically stable.

Multiplying equations \eqref{eq2:1} consequently by $\psi_k(\textbf{x}),\,k=1,2,\ldots$ and retaining only linear terms, after substituting \eqref{eq3:5} into \eqref{eq2:1}, we obtain
\begin{equation}\label{eq3:9}
    \frac{dc_k^i(t)}{dt}=\hat{v}_i(\Vect{Ac_k})_i-\lambda_kd_ic_k^i(t),\quad k=1,2,\ldots,\quad i=1,\ldots,n,
\end{equation}
where $\Vect{c_k}=(c_k^1(t),\ldots,c_k^n(t)),\,\lambda_k$ are the eigenvalues of \eqref{eq3:2}. Linear system \eqref{eq3:9} has the matrix $\textbf{R}_k$ with the elements
$$
r_{ij}^k=a_{ij}\hat{v}_i-\lambda_kd_i\delta_{ij},
$$
where $\delta_{ij}$ is the Kronecker symbol. From \eqref{eq3:8a} and the last expressions we have that
$$
r_{ij}^k=q_{ij}+(\Vect{A^\tau\hat{v}})_i\hat{v}_i-\lambda_kd_i\delta_{ij}.
$$
System \eqref{eq3:8} is stable, hence $\Tr \textbf{Q}=\sum_{i=1}^nq_{ii}<0$. On the other hand
$$
\Tr \textbf{R}_k=\Tr \textbf{Q}+\langle \Vect{A\hat{v}},\,\Vect{\hat{v}}\rangle-\lambda_k\sum_{i=1}^nd_i<\beta-\lambda_1\sum_{i=1}^nd_i,
$$
therefore $\Tr \textbf{R}_k$ is negative if \eqref{eq3:4} holds, which completes the proof.
\end{proof}

\begin{remark}Theorem \ref{th3:1} gives only a necessary condition for the homogeneous rest point to be asymptotically stable. It is possible to specify the necessary and sufficient conditions for a particular case $n=2$. In this case we have
$$
\Vect{A}=\begin{pmatrix}
           a_{11} & a_{12} \\
           a_{21} & a_{22} \\
         \end{pmatrix}.
$$
System \eqref{eq1:1} has an asymptotically stable equilibrium $\Vect{\hat{v}}=(\hat{v}_1,\,\hat{v}_2)\in \Int S_2$ if and only if
$$
a_{11}<a_{21},\quad a_{22}<a_{12}.
$$
Denoting $c_1=a_{11}-a_{21}<0,\,c_2=a_{22}-a_{12}<0,\,\Delta=a_{11}a_{22}-a_{12}a_{21}$, we have
$$
\hat{v}_1=\frac{c_2}{c_1+c_2}\,,\quad \hat{v}_2=\frac{c_1}{c_1+c_2}\,.
$$
System \eqref{eq3:9} has the matrix
$$
\Vect{R}_k=\begin{pmatrix}
           a_{11}\hat{v}_1-\lambda_kd_1 & a_{12}\hat{v}_1 \\
           a_{21}\hat{v}_2 & a_{22}\hat{v}_2-\lambda_kd_2 \\
         \end{pmatrix}.
$$
For the equilibria to be stable we need first that
$
\Tr \Vect{R}_k<0,
$
which is equivalent to
\begin{equation}\label{eqN:1}
    \sum_{i=1}^2d_i>\frac{1}{\lambda_1}\left(\beta-\frac{c_1c_2}{c_1+c_2}\right).
\end{equation}
Condition \eqref{eqN:1} is a generalization of \eqref{eq3:4} since $c_1c_2/(c_1+c_2)<0$.

The second condition is $\det \Vect{R}_k>0$, or
\begin{equation}\label{eqN:2}
    d_1d_2\lambda_k^2-\lambda_k(d_1a_{22}\hat{v}_2+d_2a_{11}\hat{v}_1)+\Delta\hat{v}_1\hat{v}_2>0.
\end{equation}
If $(d_1a_{22}\hat{v}_2+d_2a_{11}\hat{v}_1)^2>4d_1d_2\Delta\hat{v}_1\hat{v}_2$ then \eqref{eqN:2} is satisfied. The last inequality is equivalent to
\begin{equation}\label{eqN:3}
    \Delta>0,\quad a_{12}a_{21}\leq 0.
\end{equation}

In the case $\Delta<0$ the condition for the replicator equation \eqref{eq2:1} to have asymptotically stable equilibrium is that the largest root of the quadratic equation
$$
d_1d_2\lambda^2-\lambda(d_1a_{22}\hat{v}_2+d_2a_{11}\hat{v}_1)+\Delta\hat{v}_1\hat{v}_2=0
$$
should satisfy the condition
\begin{equation}\label{eqN:4}
    \lambda^*<\lambda_1.
\end{equation}
Therefore the necessary and sufficient conditions for the homogeneous rest point $\Vect{\hat{u}}\in \Int S_2(\Omega)$ be asymptotically stable are the conditions \eqref{eqN:1} and \eqref{eqN:3} or the conditions \eqref{eqN:1} and \eqref{eqN:4}.

Consider for example the simplest hypercycle equation with
$$
\Vect{A}=\begin{pmatrix}
           0 & a_{12} \\
           a_{21} & 0 \\
         \end{pmatrix}.
$$
Here we have that condition \eqref{eqN:1} is always satisfied. The condition $\det \Vect{R}_k>0$ can be written as
$$
d_1d_2\lambda^2_k>a_{12}a_{21}\hat{v}_1\hat{v}_2.
$$
If this condition does not hold then, as was shown in \cite{bratus2009existence} for the general $n$-dimensional case, the inner rest point becomes unstable.

\end{remark}

If $\Omega$ is one-dimensional, then for the cases $(\Vect{Au})_i=k_iu_i$ and $(\Vect{Au})_i=k_iu_{i-1}$ there exist non-uniform steady state solutions to \eqref{eq2:5}--\eqref{eq2:7} (see \cite{bnp2010,bratus2009existence}). More precisely, a necessary condition for such a solution to exist is $$d=\sum_{i=1}^n\frac{d_i}{k_i}<\frac{1}{\pi^2}\,.$$ We note that $\pi^2$ is the first non-zero eigenvalue of \eqref{eq3:2} when $\Omega=[0,1]$. Remark that this condition is a particular case of the general condition \eqref{eq3:4}, and we conjecture that the same situation occurs in the general case.

The natural question is whether non-uniform stationary solutions appear in the general case~\eqref{eq2:1}. Here we show that, at least for some matrices $\Vect{A}$, system \eqref{eq2:1} possesses non-uniform steady state solutions.

We rewrite \eqref{eq2:5} in the form
\begin{equation}\label{eq3:10}
    \begin{split}
      \frac{dw_i}{dx} &=p_i, \\
       d_i\frac{dp_i}{dx} & =-w_i\left[(\Vect{Aw})_i-f^{sp}\right],
    \end{split}
\end{equation}
for any $i=1,\ldots,n$. The initial data are
\begin{equation}\label{eq3:11}
    p_i(0)=w_i'(0)=0.
\end{equation}
We assume that matrix $\Vect{A}$ is such that $\Vect{Aw}>0$ for any $\Vect{w}>0$.

System \eqref{eq3:10} is conservative, its rest points $(\Vect{\hat{w}},\,\Vect{\hat{p}})$ can be found from
$$
\hat{p}_i=0,\quad (\Vect{A\hat{w}})_i=\beta,\quad \beta=\langle\Vect{A\hat{w}},\,\Vect{\hat{w}}\rangle,\quad i=1,\ldots,n,
$$
therefore $\Vect{\hat{w}}$ is a rest point of \eqref{eq1:1}. Consider the Jacobi matrix of \eqref{eq3:10} evaluated at the rest point $(0,\Vect{\hat{w}})$:
$$\textbf{J}=\left(
  \begin{array}{cc}
    \Vect{0} & \Vect{I} \\
    -\textbf{J}_d^l & \Vect{0} \\
  \end{array}
\right),
$$
where $\Vect{0}$ is the $n\times n$ zero matrix, $\Vect{I}$ is the $n$-dimensional identity matrix, and $\textbf{J}_d^l$ is the Jacobi matrix of \eqref{eq1:1} at the rest point $\Vect{\hat{v}}=\Vect{\hat{w}}$ when $i$-th row divided by $d_i$ for each row. The eigenvalues of \textbf{J} are the roots of the equation
$$
\lambda^{2n}+\det \textbf{J}_d^l=0,\quad \det \textbf{J}_d^l=(d_1\cdot\ldots\cdot d_n)^{-1}\det \textbf{J}^l,
$$
where $\det \textbf{J}^l$ is the determinant of the Jacobi matrix of \eqref{eq1:1} at $\Vect{\hat{v}}$. if $\det \textbf{J}^l<0$ then for any $n\geq 2$ the last equation has pure imaginary roots. If $\det \textbf{J}^l>0$ then the same is possible when $n=2k+1,\,k=1,2,\ldots$

Let us introduce the following sets in the phase space:
\begin{equation*}
    \begin{split}
       \Sigma & =\{\Vect{p}\in \mathbb R^n\colon p_i=0,\,i=1,\ldots,n\}, \\
        U^-&=\{\Vect{w}\in S_n\colon (\Vect{Aw})_i-{f^{sp}}<0,\,i=1,\ldots,n\},\\
        U^+&=\{\Vect{w}\in S_n\colon (\Vect{Aw})_i-{f^{sp}}>0,\,i=1,\ldots,n\},\\
        \Pi &=\{\Vect{w}\in S_n\colon w_i=\hat{v}_i,\,i=1,\ldots,n\}.
     \end{split}
\end{equation*}

\begin{figure}
\centering
\includegraphics[width=0.5\textwidth]{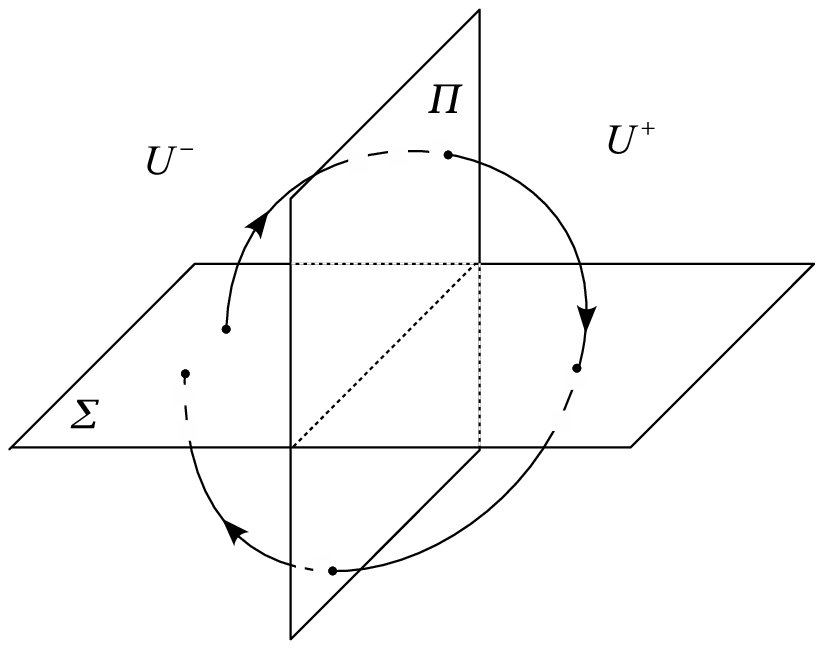}
\caption{}\label{fig3:1}
\end{figure}

From \eqref{eq3:11} it follows that at the initial ``time'' the orbit of the system \eqref{eq3:10} belongs to $\Sigma$. Suppose that $\Vect{w}(0)\in U^-$. Then the second equation in \eqref{eq3:10} yields that functions $p_i(x)$ increase as $x$ increases, and therefore $p_i(x)>0,\,x>0,\,i=1,\ldots,n$. From the condition on $\Vect{A}$ it follows that $(\Vect{Aw}(0))_i<(\Vect{Aw}(x))_i,\,x>0$ and therefore the values $(\Vect{Aw}(x))_i-{f^{sp}}$ decrease as $x$ increases, hence $\Vect{w}(x)$ has to cross the hyperplane $\Pi$ where $(\Vect{Aw}(x))_i={f^{sp}},\,i=1,\ldots,n$. After that, the solution gets into the set $U^+$, which implies that $p_i(x),\,i=1,\ldots,n$ are decreasing and $\Vect{w}(x)$ are increasing. Therefore there exists such value $x_i^*$ such that $p_i(x_i^*)=w'_i(x_i^*)=0,\,i=1,\ldots,n$ (see Fig. \ref{fig3:1}).

The diffusion coefficients $d_i$ characterize $(n+i)-$th component of the movement speed along the phase trajectories. An increase (decrease) in values $d_i$ corresponds to the decrease (increase, respectively) in the movement speed of $(n+i)-$th component of the speed vector. Therefore, it is possible to find such values of $d_i,\,i=1,\ldots,n$ such that all $x_i^*=1$.

We remark that if we reduce values $d_i$ twice this would mean that the phase orbit again would reach $\Pi$ (one cycle). Reducing $d_i$ four times, we obtain the orbit that makes two cycles in the state space, and so on. Therefore, from the discussion above, it follows that system \eqref{eq2:5} can have non-uniform stationary solutions which may possess an arbitrary number of oscillations (see a similar discussion in \cite{bratus2009existence}, where some examples are given).

\section{Dynamics of the distributed replicator equation}\label{sec4}
\begin{definition}\label{def4:1}
We shall say that the vector function $\Vect{\hat{w}}(\Vect{x})\in S_n(\Omega)$ is a distributed Nash equilibrium if
\begin{equation}\label{eq4:1}
    \int_{\Omega} \langle \Vect{u}(\Vect{x},t),\,\Vect{A\hat{w}}(\Vect{x})\rangle\, d\Vect{x}\leq \int_{\Omega}\langle\Vect{\hat{w}}(\Vect{x}),\,\Vect{A\hat{w}}(\Vect{x})\rangle \,d\Vect{x}
\end{equation}
for any vector-function $\Vect{u}(\Vect{x},t)\in S_n(\Omega_t),\,\Vect{u}(\Vect{x},t)\neq \Vect{\hat{w}}(\Vect{x})$.
\end{definition}

\begin{remark} Let $\Vect{\hat{w}}\in S_n$ satisfy \eqref{eq4:1}. Then $\Vect{\hat{w}}$ is a Nash equilibrium in the game with payoff $\Vect{A}$. Indeed,
$$
 \int_{\Omega} \langle \Vect{u}(\textbf{x},t),\,\Vect{A\hat{w}}\rangle\, d\textbf{x}=\langle \Vect{\bar{u}}(t),\,\Vect{A\hat{w}}\rangle,\quad \bar{u}_i(t)=\int_{\Omega} u_i(\textbf{x},t)\,d\textbf{x}.
$$ \eqref{eq2:3} implies that $\Vect{\bar{u}}(t)\in S_n$ for any $t$. Therefore we have
$$\langle \Vect{\bar{u}},\,\Vect{A\hat{w}}\rangle\leq \langle \Vect{A\hat{w}},\,\Vect{\hat{w}}\rangle,
$$
which means that $\Vect{\hat{w}}$ is a Nash equilibrium and therefore a rest point of \eqref{eq1:1}.

\end{remark}

\begin{theorem}
If $\Vect{\hat{w}(x)}\in\Int S_n(\Omega)$ is a Lyapunov stable stationary solution to \eqref{eq2:1} then $\Vect{\hat{w}}(\Vect{x})$ is a distributed Nash equilibrium.
\end{theorem}
\begin{proof}
Due to the fact that $\Vect{\hat{w}}(\textbf{x})$ is Lyapunov stable, then, for any initial data from a neighborhood $U^\delta$ of $\Vect{\hat{w}}(\Vect{x})$ in the space $W_2^k(\Omega)$, the corresponding solution satisfies~ \eqref{eq3:1}.

Suppose that $\Vect{\hat{w}}(\Vect{x})$ is not a distributed Nash equilibrium. Then, using continuity of the scalar product, there exists an index $i$ and constant $\xi>0$ such that
\begin{equation}\label{eq4:2}
    \int_{\Omega}(\Vect{Au}(\Vect{x},t))_i\,d\Vect{x}-\int_{\Omega}\langle \Vect{Au}(\Vect{x},t),\,\Vect{u}(\Vect{x},t)\rangle\,d\textbf{x}>\xi
\end{equation}
for all $\Vect{u}(\textbf{x},t)\in\Int S_n(\Omega_t)$ in a neighborhood of $\Vect{\hat{w}}(\Vect{x})\in \Int S_n(\Omega)$.

Let $\Vect{u}(\textbf{x},t)$ be a solution to \eqref{eq2:1}. Then from \eqref{eq2:1} we obtain
\begin{equation}\label{eq4:3}
    \frac{d}{dt}\overline{\ln u_i}(t)=\int_{\Omega} \bigl[(\Vect{Au}(\Vect{x},t))_i-\langle \Vect{Au}(\Vect{x},t),\,\Vect{u}(\Vect{x},t)\rangle\bigr]\,d\Vect{x}+d_i\int_{\Omega}\frac{\Delta u_i(\Vect{x},t)}{u_i(\Vect{x},t)}\,d\Vect{x},
\end{equation}
where $\overline{\ln u_i}(t)=\int_{\Omega}\ln u_i(\textbf{x},t)\, d\textbf{x}$. Using \eqref{eq2:2} we obtain
\begin{equation}\label{eq4:4}
    \int_\Omega \frac{\Delta u_i(\Vect{x},t)}{u_i(\Vect{x},t)}\,d\Vect{x}=\sum_{k=1}^m\int_{\Omega}\frac{1}{u_i^2(\Vect{x},t)}\left(\frac{\partial u_i(\Vect{x},t)}{\partial x_k}\right)^2d\Vect{x}\geq 0.
\end{equation}
With the help of \eqref{eq4:2} and \eqref{eq4:4}, it follows from \eqref{eq4:3} that
$$
\frac{d}{dt}(\overline{\ln u_i}(t))>\xi>0,
$$
and hence
$$
\overline{\ln u_i}(t)>\xi t+k_i,\quad i=1,\ldots,n.
$$

As $u_i(\Vect{x},t)>0$, and $\ln$ is a convex function, the integral Jensen's inequality implies that
\begin{equation}\label{eq4:5}
    \overline{\ln u_i}(t)\leq \ln \bar{u}_i(t).
\end{equation}
Therefore,
$$
\bar{u}_i(t)>C_0\exp\{\xi t\},\quad t\geq 0,
$$
which is impossible because $\Vect{\hat{w}}(\Vect{x})$ is a Lyapunov stable solution to \eqref{eq2:1}.
\end{proof}

\begin{definition}
We shall say that $\Vect{\hat{w}}(\Vect{x})\in S_n(\Omega)$ is a distributed evolutionary stable state \emph{(}DESS\emph{)} if
\begin{equation}\label{eq4:6}
\int_{\Omega}\langle \Vect{\hat{w}}(\Vect{x}),\,\Vect{Au}(\Vect{x},t)\rangle d\Vect{x}>\int_{\Omega}\langle \Vect{u}(\Vect{x},t),\,\Vect{Au}(\Vect{x},t)\rangle\,d\Vect{x}
\end{equation}
for any $\Vect{u}(\Vect{x},t)\in S_n(\Omega_t)$ from a neighborhood of $\Vect{\hat{w}}(\Vect{x})\in S_n(\Omega)$, $\Vect{u}(\Vect{x},t)\neq \Vect{\hat{w}}(\Vect{x})$.
\end{definition}

Let us introduce Definition \ref{def3:2}.

\begin{definition}\label{def3:2}
Stationary solution $\Vect{\hat{w}}(\Vect{x})\in S_n(\Omega)$ to \eqref{eq2:5} is stable in the sense of the mean integral value if for any $\varepsilon>0$ there exists $\delta>0$ such that for the initial data $\varphi_i(\bf{x})$ of system \eqref{eq2:1}, which satisfy
$$
|\bar{\varphi}_i-\hat{\bar{{w}}}_i|<\delta,\quad i=1,\ldots,n,
$$
where
$$
\bar{\varphi}_i=\int_\Omega\varphi_i(\Vect{x})\,d\Vect{x},\quad \hat{\bar{w}}_i=\int_\Omega \hat{w}_i(\Vect{x})\,d\Vect{x},
$$
it follows that
$$
|\bar{u}_i(t)-\hat{\bar{w}}_i|<\varepsilon
$$
for any $i=1,\ldots,n$ and $t>0$.

Here $\bar{u}_i(t)=\int_\Omega u_i(\Vect{x},t)\,d\Vect{x}$ and $u_i(\Vect{x},t),\,i=1,\ldots,n$ are the solutions to \eqref{eq2:1}, $\Vect{u}(\Vect{x},t)\in S_n(\Omega_t)$.
\end{definition}

From Definitions \ref{def3:1} and \ref{def3:2} it follows that stability in the mean integral sense is weaker than Lyapunov stability. For instance, consider a simple example: let $g(x,t)\in W_2^1,\,x\in[0,1]$ and can be represented as
$$
g(x,t)=c_0(t)+\sum_{k=1}^\infty c_k(t)\cos k\pi x.
$$
Let us suppose that $c_0(t)\to 0$ when $t\to\infty$. Then $\bar{g}(t)=\int_0^1 g(x,t)\,dx\to 0$ when $t\to\infty$, whereas $\|g(x,t)\|_{W_2^1}=\sum_{k=1}^\infty c_k^2(t)(1+k^2\pi^2)$ does not necessarily tend to zero.

\begin{theorem}\label{th4:2}
Let $\Vect{\hat{w}}\in\Int S_n$ be a spatially homogeneous solution to \eqref{eq2:5} \emph{(}i.e., $\Vect{\hat{w}}\in\Int S_n$ is a rest point of \eqref{eq1:1}\emph{)}. If $\Vect{\hat{w}}$ is DESS then $\Vect{\hat{w}}$ is an asymptotically stable solution to \eqref{eq2:1} in the sense of the mean integral value.
\end{theorem}
\begin{proof}
Consider the set of function $\Vect{u}(\Vect{x},t)\in S_n(\Omega_t)$ belonging to a neighborhood $U^\delta$ of $\Vect{\hat{w}}$ in the space $B(\Omega_t)$. Define the functional
\begin{equation*}
    V(\Vect{u}(t))=\sum_{i=1}^n \hat{w}_i\int_{\Omega}\ln u_i(\Vect{x},t)\, d\Vect{x}=\sum_{i=1}^n\hat{w}_i\overline{\ln u_i}(t)\,.
\end{equation*}
We can always choose $U^\delta$ such that $u_i(\Vect{x},t)>0$ for all $i$ because $\hat{w}_i>0$. Using \eqref{eq4:4} we obtain
\begin{equation*}
    \frac{d\,\overline{\ln u}_i}{dt}\geq \int_{\Omega}(\Vect{Au})_i\,d\Vect{x}-\int_{\Omega}\langle \Vect{Au},\,\Vect{u}\rangle\, d\Vect{x},\quad \Vect{u}\in U^\delta,
\end{equation*}
where we suppress the dependence on $t$ and $\Vect{x}$ for simplicity.
From \eqref{eq4:6} we have that
\begin{equation*}
    \begin{split}
       \frac{dV}{dt} &= \dot{V}\geq \sum_{i=1}^n \hat{w}_i\int_{\Omega} (\Vect{Au})_i\,d\Vect{x}-\int_{\Omega}\langle \Vect{Au},\,\Vect{u}\rangle\, d\Vect{x}= \\
         & =\int_{\Omega}\langle \Vect{\hat{w}},\,\Vect{Au}\rangle\, d\Vect{x}-\int_{\Omega}\langle \Vect{Au},\,\Vect{u}\rangle\, d\Vect{x}>0,\quad \Vect{u}\in U^\delta.
     \end{split}
\end{equation*}
The functional $V$ is bounded. Indeed, applying Jensen's inequality for sums
$$
\sum_{i=1}^n p_i\ln q_i\leq \ln \left(\sum_{i=1}^np_iq_i\right),\quad \sum_{i=1}^np_i=1,\quad q_i>0,
$$
we obtain
$$
\int_{\Omega}\sum_{i=1}^n\hat{w}_i\ln\frac{u_i}{\hat{w}_i}\,d\Vect{x}\leq \ln \int_{\Omega}\sum_{i=1}^nu_i\,d\Vect{x}=\ln 1=0.
$$
Hence,
$$
V(\Vect{u})=\sum_{i=1}^n\int_{\Omega}\hat{w}_i\ln u_i(\Vect{x},t)\,d\Vect{x}=\sum_{i=1}^n\hat{w}_i\overline{\ln u_i}(t)\leq\sum_{i=1}^n\hat{w}_i\ln \hat{w}_i.
$$
Since $\dot{V}(\Vect{u})>0$ for all $\Vect{u}\in U^{\delta}$ then $V$ is a strict Laypunov functional for \eqref{eq2:1} and therefore
$$
\lim_{t\to\infty}\overline{\ln u_i}(t)=\ln\hat{w}_i,\quad i=1,\ldots,n.
$$
From the last inequality and using \eqref{eq4:5} we have
\begin{equation}\label{eq4:13}
    \lim_{t\to\infty}\bar{u}_i(t)\geq \hat{w}_i,\quad i=1,\ldots,n.
\end{equation}
If \eqref{eq4:13} is strict at least for one $i$, then
$$
\sum_{i=1}^n\lim_{t\to\infty}\bar{u}_i(t)>\sum_{i=1}^n\hat{w}_i=1
$$
which is impossible since
$$
\sum_{i=1}^n\bar{u}_i(t)=1
$$
for any $t\geq 0$. This proves the theorem.
\end{proof}

\begin{remark}Let the conditions of Theorem \ref{th4:2} be met, then
\begin{equation}\label{eq4:8}
    \lim_{t\to\infty}f^{sp}(t)\geq \lim_{t\to\infty}f^{loc}(t)=\langle \Vect{\hat{w}},\Vect{A\hat{w}}\rangle.
\end{equation}
Indeed, in this case
$$
\lim_{t\to\infty}\int_{\Omega}\langle\hat{\Vect{w}},\Vect{Au}(\Vect{x},t)\rangle d\Vect{x}=\lim_{t\to\infty}\langle\Vect{\hat{w}},\Vect{A\bar{u}}\rangle=\langle\Vect{\hat{w}},\Vect{A\hat{w}}\rangle.
$$

Using \eqref{eq4:6}  we obtain \eqref{eq4:8}.
\end{remark}
\section{Sufficient conditions for DESS}\label{sec5}
Application of the results obtained so far for particular distributed replicator systems reduces to the problem of checking conditions for DESS for spatially homogeneous stationary solution $\hat{\Vect{w}}\in\Int S_n$

Let us introduce the following function
\begin{equation}\label{eq5:1}
    M(t)=\int_{\Omega}\langle\Vect{u}(\Vect{x},t),\Vect{Au}(\Vect{x},t)\rangle\, d\Vect{x}-\int_{\Omega}\langle\Vect{\hat{w}},\Vect{A{u}}(\Vect{x},t)\rangle\, d\Vect{x}.
\end{equation}
If the stationary point $\Vect{\hat{w}}\in\Int S_n$ meets the conditions for DESS then $M(t)<0$ in a neighborhood of $\Vect{\hat{w}}$ in the space $W_2^k(\Omega)$ for functions $\Vect{u}(\Vect{x},t)\in S_n(\Omega_t)$ from neighborhood of $\hat{\Vect{w}}(\Vect{x})$ in the space $S_n(\Omega_t)$, such that $\Vect{u}(\Vect{x},t)\neq\Vect{\hat{w}}(\Vect{x}),\,\Vect{x}\in\Omega,\,t\geq 0$.

Let
\begin{equation}\label{eq5:2}
    u_i(\Vect{x},t)=\hat{w}_i+c_i^0(t)+\sum_{s=1}^\infty c_i^s(t)\psi_s(\Vect{x}),\quad i=1\ldots,n,
\end{equation}
where $\psi_s(\Vect{x}),\,s=1,2,\ldots$ are the eigenfunctions of problem \eqref{eq3:2}. Due to \eqref{eq2:3}
\begin{equation}\label{eq5:3}
    \sum_{i=1}^nc_i^0(t)=0.
\end{equation}
Functions $c_i^0(t)$ are not equal to zero for $t>0$ simultaneously since $\bar{\Vect{u}}(t)\neq \Vect{\hat{w}}$. Additionally, from
$$
0\leq \int_{\Omega}u_i(\Vect{x},t)\,d\Vect{x}=\bar{u}_i(t)=\hat{w}_i+c_i^0(t)\leq 1,
$$
it follows that
$$
-\hat{w}_i\leq c_i^0(t)\leq 1-\hat{w}_i,\quad i=1,\ldots,n,
$$
i.e., these functions are bounded, which implies that there exists $\delta>0$ such that
\begin{equation}\label{eq5:4}
|\Vect{c}^0(t)|^2=\sum_{i=1}^n|c_i^0(t)|^2\geq \delta^2>0.
\end{equation}
Let us fix $\varepsilon>0$ and consider an $\varepsilon$-neighborhood of point $\Vect{\hat{w}}$ in $W_2^k(\Omega)$. From the embedding theorems
$$
\|\Vect{u}(\Vect{x},t)-\Vect{\hat{w}}\|^2_{L_2(\Omega)}\leq K\|\Vect{u}(\Vect{x},t)-\Vect{\hat{w}}\|^2_{W_2^k(\Omega)}\leq K_1\varepsilon^2,
$$
where $K,\,K_1$ are positive constants. This means that
\begin{equation}\label{eq5:5}
    \|\Vect{u}(\Vect{x},t)-\Vect{\hat{w}}\|^2_{L_2(\Omega)}=(c_i^0(t))^2+\sum_{s=1}^{\infty}(c_i^s(t))^2\leq K_1\varepsilon^2.
\end{equation}

\begin{theorem}\label{th5:1}
The stationary point $\Vect{\hat{w}}\in\Int S_n$ is a DESS if
\begin{equation}\label{eq5:6}
\langle \Vect{Ac}^0(t),\Vect{c}^0(t)\rangle\leq -\gamma^2|\Vect{c}^0(t)|^2,\quad \gamma>0
\end{equation}
for any $\Vect{c}^0(t)$ satisfying \eqref{eq5:3}
\end{theorem}
\begin{proof}
After inserting \eqref{eq5:2} into \eqref{eq5:1} we obtain
$$
M(t)=\langle \Vect{Ac}^0(t),\Vect{c}^0(t) \rangle+\sum_{s=1}^\infty \langle \Vect{Ac}^s(t),\Vect{c}^s(t) \rangle-\langle \Vect{c}^0,\Vect{A\hat{w}}\rangle.
$$
Since $\Vect{\hat{w}}\in S_n$ is a stationary state, then
$$
(\Vect{A\hat{w}})_1=\ldots=(\Vect{A\hat{w}})_n=\beta,
$$
therefore, using \eqref{eq5:3},
$$
\langle \Vect{c}^0,\Vect{A\hat{w}}\rangle=\sum_{i=1}^n(\Vect{A\hat{w}})_ic_i^0=\beta\sum_{i=1}^nc_i^0(t)=0.
$$

On the other hand,
$$
|\langle \Vect{Ac}^s(t),\Vect{c}^s(t)\rangle|\leq K_2|\Vect{c}^s(t)|^2,\quad K_2>0,
$$
and, using \eqref{eq5:6}, we obtain
$$
M(t)<-\gamma^2|\Vect{c}^0(t)|^2+K_2\sum_{s=1}^\infty|\Vect{c}^s(t)|^2.
$$
Using \eqref{eq5:4} and \eqref{eq5:5} we arrive at the estimate
$$
M(t)<(-\gamma^2\delta^2+K_1K_2\varepsilon^2).
$$
If $\varepsilon<\gamma\delta/\sqrt{K_1K_2}$, then $M(t)<0$.
\end{proof}

\begin{remark} In the case of replicator equations with symmetric $\Vect{A}$ condition \eqref{eq5:6} means that  each orbit converges to the stationary point $\Vect{\hat{w}}\in\Int S_n$ \cite{hofbauer1998ega}.\end{remark}

\begin{example}
Consider system \eqref{eq2:1} with the matrix
$$
\Vect{A}=\left(
           \begin{array}{cc}
             a & b \\
             c & d \\
           \end{array}
         \right),\quad a<c,\,d<b.
$$
We have
$$
\langle \Vect{Ac}^0(t),\Vect{c}^0(t)\rangle=-(c_1^0)^2(b-a+c-d),
$$
and hence condition \eqref{eq5:6} is satisfied if $\gamma^2=b-a+c-d>0$.
\begin{figure}
\centering
\includegraphics[width=0.99\textwidth]{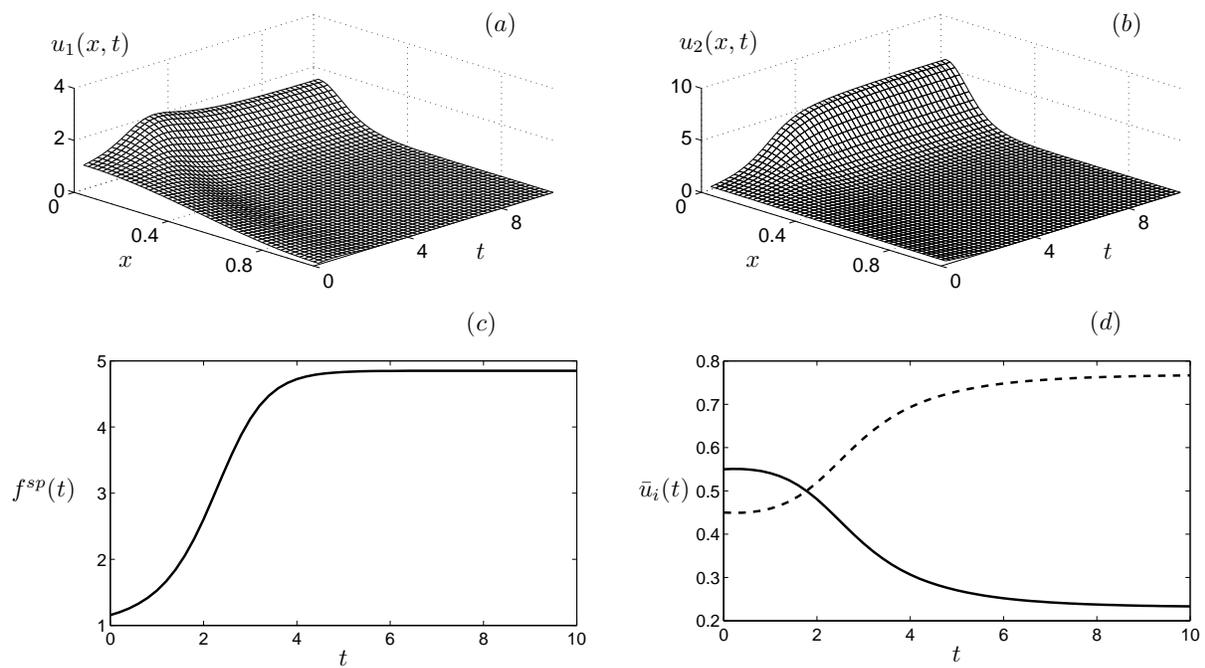}
\caption{Solution to \eqref{eq2:1}, see Example 1 and text for details. $(a),(b)$ Time dependent solutions; $(c)$ Time evolution of $f^{sp}(t)$; $(d)$ $\bar{u}_i(t)=\int_0^1 u_i(x,t)\,dx,\,i=1,2$ are shown}\label{fig5:1}
\end{figure}

In Fig. \ref{fig5:1} some numerical calculations are presented for the case when
$$
\Vect{A}=\left(
           \begin{array}{cc}
             0.8 & 1.1 \\
             1.2 & 0.9 \\
           \end{array}
         \right).
$$
The diffusion coefficients are $d_1=0.03,\,d_2=0.02$. In this case \eqref{eq3:4} is not satisfied and there are spatially heterogeneous stationary solutions to the system \eqref{eq2:1}. The dynamics of the solutions is shown in Fig. \ref{fig5:1}$(a),\,(b)$. Spatially non-uniform solutions are stable in this case.
\begin{figure}[!b]
\centering
\includegraphics[width=0.95\textwidth]{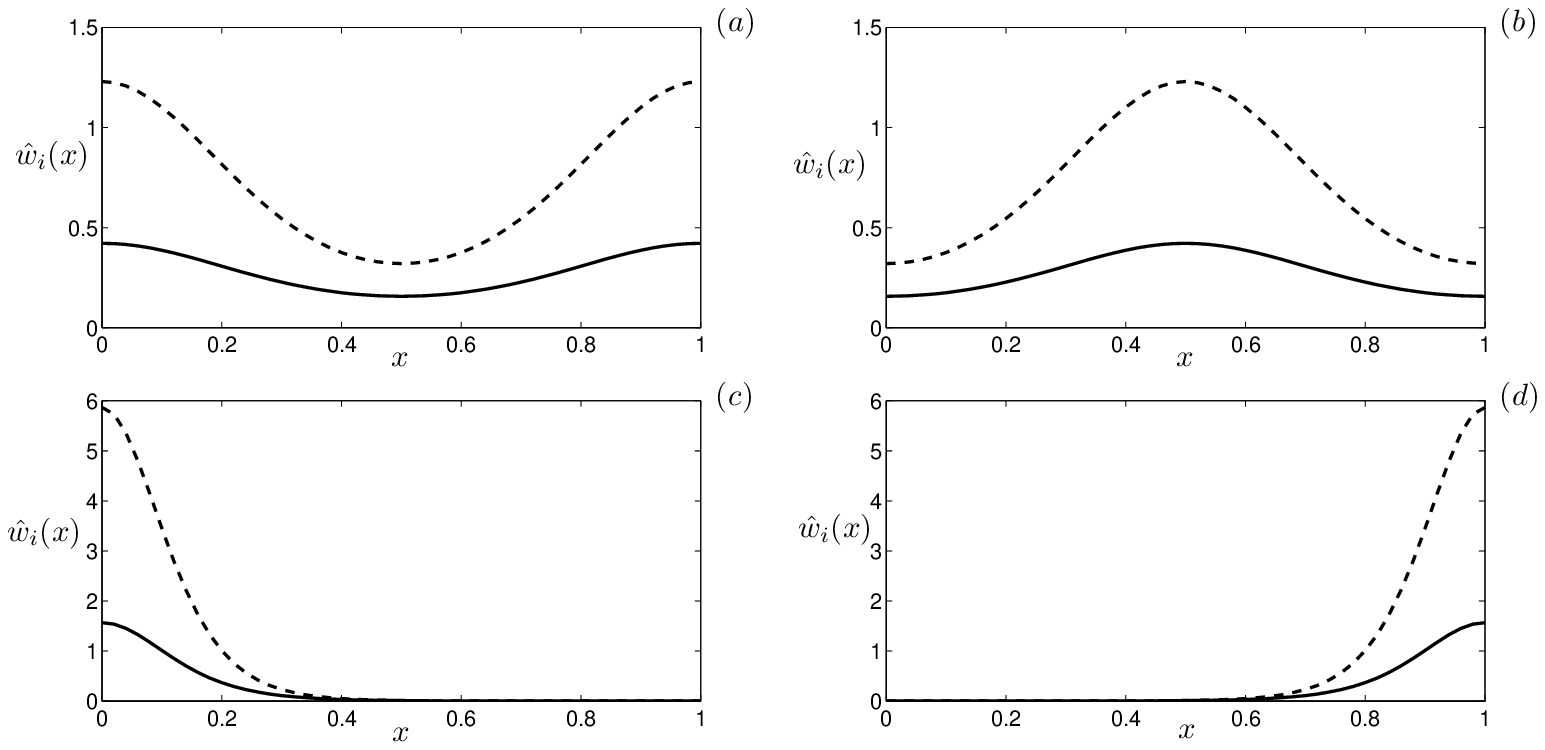}
\caption{Asymptotically stable stationary non-uniform solutions to the problem as in Example~1. See text for details. The numerical scheme is presented in \cite{bratus2006ssc}}\label{fig5:2}
\end{figure}

In Fig. \ref{fig5:2} all possible stable spatially non-uniform solutions are presented, each of which has its own basin of attraction.
\end{example}

\begin{example}
$$
\Vect{A}=\left(
  \begin{array}{ccc}
    \mu & 1 & -1 \\
    -1 & \mu & 1 \\
    1 & -1 & \mu \\
  \end{array}
\right).
$$

In this case,
$$
\langle \Vect{Ac}^0(t),\Vect{c}^0(t)\rangle=\mu\sum_{i=1}^3|c_i^0|^2,
$$
and \eqref{eq5:6} holds if $\mu<0$.
\end{example}

\begin{example}[Hypercycle equation]
Consider a hypercycle system with three members and the matrix
$$
\Vect{A}=\left(
  \begin{array}{ccc}
    0 & 0 & 1 \\
    1 & 0 &0 \\
    0 & 1 & 0 \\
  \end{array}
\right).
$$
From \eqref{eq5:3} $c_1^0=-(c_2^0+c_3^3)$, therefore

\begin{equation*}
\begin{split}
\langle \Vect{Ac}^0(t),\Vect{c}^0(t)\rangle&=-(c_2^0+c_3^3)c_3^0-c_2^0(c_2^0+c_3^3)+c_3^0c_2^0\\
&=-[(c_3^0)^2+c_2^0c_2^0+(c_2^0)^2]<-\frac 12 [(c_3^0)^2+(c_2^0)^2],
\end{split}
\end{equation*}
which yields \eqref{eq5:6}.
\end{example}
\begin{example}\label{ex5:3}
$$
\Vect{A}=\left(
           \begin{array}{ccc}
             a & b & c \\
             c & a & b \\
             b & c & a \\
           \end{array}
         \right).
$$
Matrix $\Vect{A}$ is circulant, and its eigenvalues are
$$
\lambda_1=a+b+c,\quad \lambda_{2,3}=(a-(b+c)/2)\pm i\sqrt{3}(b+c)/2.
$$
Assume $a<(b+c)/2, b>a>c>0,\,bc>a^2$, then $\Re \lambda_{2,3}<0$. $\lambda_1$ has the eigenvector $(1,1,1)$. This vector is orthogonal to any $\Vect{c}^0$ that satisfies \eqref{eq5:3}, hence $\Vect{A}$ is negatively determined on this set.
\end{example}

\begin{example}
$$
\Vect{A}=\left(
           \begin{array}{cccc}
             0 & a_1 & a_2 & a_3 \\
             a_3 & 0 & a_1 & a_2 \\
             a_2 & a_3 & 0 & a_1 \\
             a_1 & a_2 & a_3 & 0 \\
           \end{array}
         \right).
$$
Let $a_1+a_3>a_2$. The eigenvalues of $\Vect{A}$ are given by
\begin{equation*}
    \begin{split}
       \lambda_1 & =a_1+a_2+a_2 \\
        \lambda_2 & =-a_1-a_2-a_3\\
        \lambda_{3,4}&=-a_2\pm i(a_1+a_3)
     \end{split}
\end{equation*}
In the same vein as in Example \ref{ex5:3}, $\lambda_1$ corresponds to the eigenvector $\Vect{u}=(1,1,1,1)$, which is orthogonal to any $\Vect{c}^0$ satisfying \eqref{eq5:3}. The other eigenvalues have negative real parts, therefore $\Vect{A}$ is negatively determined.

Consider the following matrix
$$\Vect{A}\left(
  \begin{array}{cccc}
    0 & 0.5 & 1 & 0.8 \\
    0.8 & 0 & 0.5 & 1 \\
    1 & 0.8 & 0 & 0.5 \\
    0.5 & 1 & 0.8 & 0 \\
  \end{array}
\right)
$$
and the diffusion coefficients $\Vect{d}=(0.03,0.02,0.03,0.02)$. Again we have spatially non-uniform solutions which are stable in the mean integral sense (see Fig.~\ref{fig5:3}). In Fig.~\ref{fig5:5} possible asymptotically stable non-uniform solutions are shown. The details of the numerical scheme are given in \cite{bratus2006ssc}.
\begin{figure}
\centering
\includegraphics[width=0.95\textwidth]{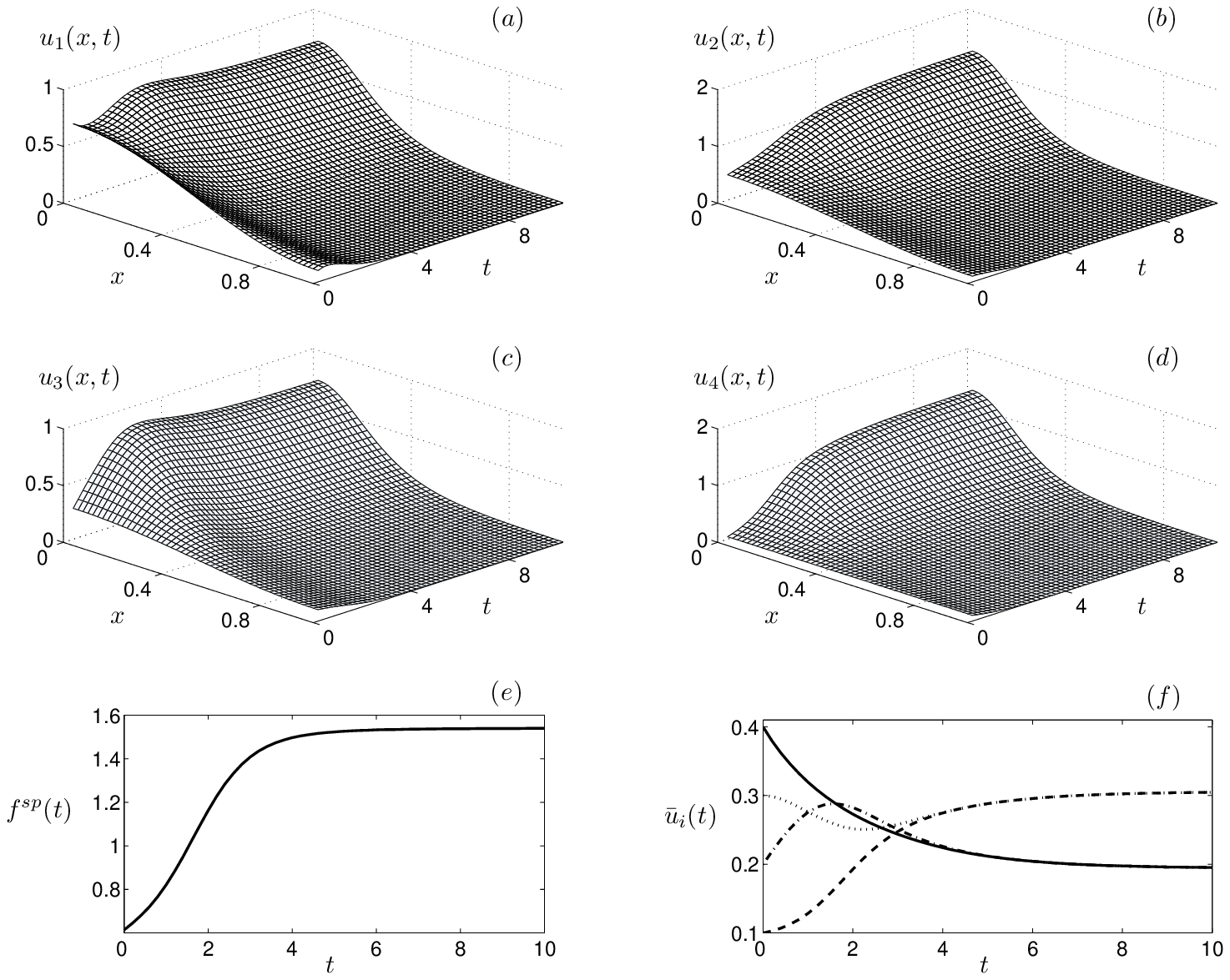}
\caption{Solutions to \eqref{eq2:1}, see Example 5 and text for details. $(a),(b),(c),(d)$ Time dependent solutions; $(e)$ Time evolution of $f^{sp}(t)$; $(f)$ $\bar{u}_i(t)=\int_0^1 u_i(x,t)\,dx,\,i=1,\ldots,4$ are shown}\label{fig5:3}
\end{figure}
\end{example}

\begin{figure}[!t]
\centering
\includegraphics[width=0.95\textwidth]{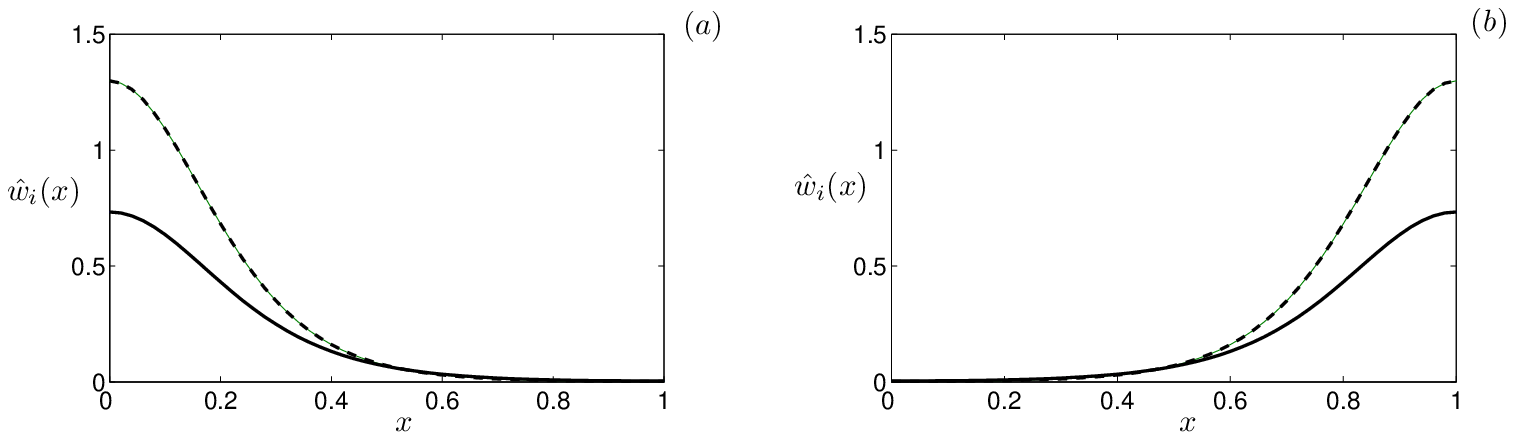}
\caption{Asymptotically stable stationary non-uniform solutions to the problem of Example~5. See text for details. The solutions are such that $\hat{w}_1(x)=\hat{w}_3(x)$ (dashed curve), and $\hat{w}_2(x)=\hat{w}_4(x)$ (solid curve), ${f}^{sp}=1.54$}\label{fig5:5}
\end{figure}

Here the conditions for the spatially homogeneous solutions to be DESS are fulfilled, therefore, this solution is stable in the sense of the mean integral value, which can be seen from the pictures.

\paragraph{Acknowledgments:} The research is supported in part by the Russian Foundation for Basic Research grant \# 10-01-00374. ASN is supported by the grant to young researches from Moscow State University of Railway Engineering.

\end{document}